\documentclass[conference]{IEEEtran}
\IEEEoverridecommandlockouts

\usepackage{graphicx}
\usepackage{amsmath}
\usepackage{amssymb}
\usepackage{amsthm}
\usepackage{color}
\usepackage{bm}
\usepackage{comment}
\usepackage{cite}
\usepackage{url}
\usepackage{array}
\usepackage{multicol}
\usepackage{subcaption}

\usepackage{todonotes}
\usepackage{scalerel,stackengine}
\stackMath
\newcommand\reallywidehat[1]{%
\savestack{\tmpbox}{\stretchto{%
  \scaleto{%
    \scalerel*[\widthof{\ensuremath{#1}}]{\kern-.6pt\bigwedge\kern-.6pt}%
    {\rule[-\textheight/2]{1ex}{\textheight}}%WIDTH-LIMITED BIG WEDGE
  }{\textheight}% 
}{0.5ex}}%
\stackon[1pt]{#1}{\tmpbox}%
}

\theoremstyle{plain}
\newtheorem{proposition}{Proposition}

\theoremstyle{remark}
\newtheorem{cons}{Corollary}

\usepackage{listofitems}
\usepackage{algorithm}
\usepackage{algpseudocode}
\algnewcommand\algorithmicinput{\textbf{Input:}}
\algnewcommand\INPUT{\item[\algorithmicinput]}
\algnewcommand\algorithmicoutput{\textbf{Output:}}
\algnewcommand\OUTPUT{\item[\algorithmicoutput]}

\newcommand{\cV}{{\cal V}}

\floatname{algorithm}{Algorithm}

\algrenewcommand\algorithmiccomment[1]{\{#1\}}
\algnewcommand{\To}{\textbf{to}\ }

\usepackage{mathtools}

\graphicspath{{figures/}}

\newtheorem{definition}{Definition}

\newtheorem{theorem}{Theorem}

\newtheorem{remark}{Remark}

\newcommand{\rank}{\mathrm{rank}\kern 2pt}

\renewcommand{\epsilon}{\varepsiolon}

\def\BibTeX{{\rm B\kern-.05em{\sc i\kern-.025em b}\kern-.08em
    T\kern-.1667em\lower.7ex\hbox{E}\kern-.125emX}}

\begin{document}
\title{ Probabilistic Examination of Least Squares Error \\ in Low-bitwidth Cholesky Decomposition
\thanks{The research was carried out at Skoltech and supported by the Russian Science Foundation (project no. 24-29-00189).
}}

\author{\IEEEauthorblockN{1\textsuperscript{st} Alexander Osinsky}
\IEEEauthorblockA{
\textit{Skoltech}\\
Moscow, Russia \\
A.Osinskiy@skoltech.ru}
\and
\IEEEauthorblockN{2\textsuperscript{nd} Roman Bychkov}
\IEEEauthorblockA{
\textit{Skoltech}\\
Moscow, Russia \\
R.Bychkov@skoltech.ru}
\and
\IEEEauthorblockN{3\textsuperscript{rd} Mikhail Trefilov}
\IEEEauthorblockA{
\textit{MIPT, Phystech}\\
Moscow, Russia \\
trefilov.mp@phystech.edu} 
\and 
\IEEEauthorblockN{4\textsuperscript{th} Vladimir Lyashev}
\IEEEauthorblockA{
\textit{MIPT, Phystech}\\
Moscow, Russia \\
lyashev.va@mipt.ru}
\and
\IEEEauthorblockN{5\textsuperscript{th} Andrey Ivanov}
\IEEEauthorblockA{
\textit{Skoltech}\\
Moscow, Russia \\
AN.Ivanov@skoltech.ru}
}

% conference papers do not typically use \thanks and this command
% is locked out in conference mode. If really needed, such as for
% the acknowledgment of grants, issue a \IEEEoverridecommandlockouts
% after \documentclass

% for over three affiliations, or if they all won't fit within the width
% of the page, use this alternative format:
% 
%\author{\IEEEauthorblockN{Michael Shell\IEEEauthorrefmark{1},
%Homer Simpson\IEEEauthorrefmark{2},
%James Kirk\IEEEauthorrefmark{3}, 
%Montgomery Scott\IEEEauthorrefmark{3} and
%Eldon Tyrell\IEEEauthorrefmark{4}}
%\IEEEauthorblockA{\IEEEauthorrefmark{1}School of Electrical and Computer Engineering\\
%Georgia Institute of Technology,
%Atlanta, Georgia 30332--0250\\ Email: see http://www.michaelshell.org/contact.html}
%\IEEEauthorblockA{\IEEEauthorrefmark{2}Twentieth Century Fox, Springfield, USA\\
%Email: homer@thesimpsons.com}
%\IEEEauthorblockA{\IEEEauthorrefmark{3}Starfleet Academy, San Francisco, California 96678-2391\\
%Telephone: (800) 555--1212, Fax: (888) 555--1212}
%\IEEEauthorblockA{\IEEEauthorrefmark{4}Tyrell Inc., 123 Replicant Street, Los Angeles, California 90210--4321}}

% make the title area
\maketitle

\begin{abstract}

In this paper, we propose a new approach to justify a round-off error impact on the accuracy of the linear least squares (LS) solution using Cholesky decomposition. This decomposition is widely employed to inverse a matrix in the linear detector of the Multi-User multi-antenna receiver. 
The proposed stochastic bound is much closer to actual errors than other numerical bounds. It was tested with a half-precision format and validated in realistic scenarios. Experimental results demonstrate our approach predicts errors very close to those achieved by simulations. The proposed approach can be employed to analyze the resulting round-off error in many other applications.

%  One possible solution is to use low resolution arithmetic to reduce data transfer as well as to accelerate the computations. However, there is a little research on disturbances caused by low arithmetic processing. In this letter, we propose a new approach to justify a round-off error impact to the performance of low-precision MMSE-IRC detector based on Cholesky decomposition. Experimental results demonstrate our approach predicts errors of half-precision processing very close to achieved by simulations. The method can be applied to analyze the error for any resolution in different part of communication system.

 %\IEEEpeerreviewmaketitle
 \vskip 0.2cm


\end{abstract}

\section{Introduction}
\label{sec:introduction}

\IEEEPARstart Future generation wireless communications face the issue of an extremely high computational complexity due to the increasing number of antennas as well as the number of users in Massive Multiple Input, Multiple Output (MIMO) systems \cite{shafi20175g}. The excessive complexity of the receiver \cite{MMSEmatrixmethods} stands as a major issue. One of the operations that require the most computing is MIMO detection \cite{Studer}. While the number of antennas in a standard 4G MIMO is limited by $N^{4G}_{RX}=8$, in massive MIMO it begins at $N^{5G}_{RX}=64$ \cite{shafi20175g}. Thus, the antenna domain complexity of the linear least squares (LS) detector \cite{high_perform} in 5G is about
\[
[{N^{5G}_{RX}}/{N^{4G}_{RX}}]^3=512
\]
times greater than in 4G. Considering the Cholesky decomposition has the lowest complexity, there is an acute problem of further lowering its complexity.

A feasible method to further lower the complexity of Cholesky decomposition is to decrease its computational precision. Using, e.g. half precision, allows reducing computational cost several times compared to single precision \cite{FloatOld}. Unfortunately, the application of low precision numbers can lead to significant accuracy degradation \cite{fixed}, when the condition number of the transformation is high \cite{OsinskyRegularization}. Thus, the investigation of low precision arithmetic impact on the resulting error is rather important. 

Moreover, if the resulting round-off error of the MIMO detector \cite{cloud, high_perform} is expected to be less than, for example, the channel estimation error \cite{Theory_bound, Efficient, Lower_bound}, the detector bitwidth can be decreased.

To the best of our knowledge, there is very limited research on this topic. Thus, existing bounds of Cholesky decomposition accuracy mostly employ numerical analysis to derive an upper bound of the resulting round-off error. Such bound is derived for the worst case and much exceeds practical values. For example, it can demand us to use several bits higher mantissa length than we need, thus lowering all efforts on bitwidth minimization. As a result, the predicted error is overestimated, and the bound is not very useful in the problem we are solving.  

%\subsection{Linear Least Squires}

%In real-world applications, the number of users and the number of propagation channel taps are limited, so there is a limited set of directions of arrival at the receiver. As a result, the antenna signal transformation to a reduced dimensionality beamspace \cite{lowerbound} is widely employed to address issues of high complexity in Multi-User Massive MIMO receivers. After transformation to the beamspace domain, the received signal $Y \in \mathbb{C}^{N_{beam}}$ on a single subcarrier is given by:
%\begin{equation}\label{eq:Y}
%  Y = H X + E,
%\end{equation}
%where $N_{beam}<N_{RX}$ is the beamspace size, $H \in \mathbb{C}^{N_{beam} \times N_{TX}}$ is the channel matrix, $X \in \mathbb{C}^{N_{TX}}$ is the data symbols transmitted by $N_{TX}$ antennas for all users, $E \in \mathbb{C}^{N_{beam} \times N_{SC}}$ is the additive white Gaussian noise (AWGN) with power $\sigma^2$ per beam.

%The least squares (LS) detector can be written as:
%\begin{equation}\label{eq:detorig}
%  \hat X = \left( H^H H \right)^{-1} H^H %Y = H^+ Y,
%\end{equation}
%where $H^+$ denoted Moore-Penrose pseudoinverse of $H$. To estimate the round-off errors in computations of the detected value $\hat X$ in low-precision arithmetic, we first simplify the problem by treating $\hat X = X$ as an exact solution. Then we can compare the round-off error $\Delta X$ with the white noise error $\hat X - X$ to see which one is larger.

\subsection{Problem definition}

We are interested in solving the linear least squares (LS) problem
\begin{equation}\label{eq:HXY}
  \left\| HX - Y \right\|_2 \to \min,
\end{equation}
with $X$ being an exact solution of $HX=Y$. Here $X \in \mathbb{C}^N$ is the transmitted signal, $Y \in \mathbb{C}^M$ is the received signal, and $H \in \mathbb{C}^{M \times N}$ is the channel matrix. In practice signal $Y$ also contains noise, but here we are only interested in the round-off error: if the white noise is assumed to be independent of the round-off, we can add its corresponding error independently. The problem \eqref{eq:HXY} is equivalent to
\[
  H^H H X = H^H Y
\]
or
\[
  X = \left( H^H H \right)^{-1} H^H Y
\]
with $H$ having full column rank. 
%We will assume $\left\| X \right\|_2 = 1$, since the final error is going to be relative.

This problem can be solved by various QR decompositions like Modified Gramm-Schmidt or Householder methods. However, the fastest way is to use Cholesky decomposition $L L^H = H^H H$ or Cholesky QR, like in \cite{CholeskyQR3} (Cholesky QR is computing $Q = H L^{-H}$ with $L$ from Cholesky). After factors $Q$ and $R$ are obtained, we compute
\[
  W = H^+ = R^{-1} Q^H
\]
and then
\begin{equation}\label{eq:WY}
  X = H^+ Y = WY.
\end{equation}
Weight matrix $W$ is stored directly instead of computing $Q^H Y$ and then using back substitution, because of practical applications of this problem: if $H$ is a channel matrix, and we want to compute the LS detector for multiple right-hand sides and multiple channels, it is advantageous to do interpolation of $W$ instead of computing it every time for slightly different channels $H$.

Here we are assuming no interference scenario and neglect the noise power regularization. In general, some interference correlation matrix $R_{UU}$ should be used to account for interference, which is also decomposed using Cholesky factorization. Round-off errors due to Cholesky decomposition of $R_{UU}$ in high interference environment were already studied in \cite{OurCholeskyRoundoff}. Here we instead study the Cholesky decomposition of $H^H H$ and its round-off error, which is going to be much more important in case of low interference.

\subsection{Contribution}

In this letter, we propose a probabilistic bound for the resulting round-off error of the Cholesky-based linear LS solution implemented in low-precision arithmetic. Compared to existing numerical bounds, this bound is much closer to the practical error.

Our results may help to predict the minimum required precision for the arithmetic operations, involved in linear LS, which do not yet lead to the loss of performance. Our approach can also potentially help to predict, which parts of the algorithm should use higher precision.

\section{Known results}

The most known bound for Cholesky decomposition is the following theorem.

\begin{theorem}[\cite{accuracy}, Theorem 10.3]\label{th:CholeskyOld}
If Cholesky decomposition of $A \in \mathbb{R}^{N \times N}$ with rounding to the nearest runs to completion and outputs some factors $\tilde L$ and $\tilde L^H$, then
\[
  \tilde L \tilde L^H = A + \Delta A, 
\]  
\[ \left| \Delta A \right| \leqslant \left( N + 1 \right) u \left| L \right| \left| L^H \right| + O \left( u^2 \right),
\]
where $u = 2^{-b-1}$, $b$ is the number of mantissa bits.
\end{theorem}

Here $\left| A \right|$ denotes the matrix with absolute values of the elements of $A$. Corresponding matrix inequalities are elementwise:
\[
  \left| A \right| \leqslant \left| B \right| \Leftrightarrow \left| A_{ij} \right|^2 \leqslant \left| B_{ij} \right|^2.
\]

In spectral norm, Theorem \ref{th:CholeskyOld} leads to 
\[
\left\| \Delta A \right\|_2 \leqslant \left( N + 1 \right) N u \left\| A \right\|_2 + O \left( u^2 \right)
\]
and in the Frobenius norm to 
\[
\left\| \Delta A \right\|_F \leqslant \left( N + 1 \right) \sqrt{N} \left\| A \right\|_F + O \left( u^2 \right).
\]
Either way, the dependence on size $N$ is too significant. As we will see, it would lead to a much larger error, than observed. Our goal will be to prove a better bound (although with some less rigorous assumptions) and use it to estimate the round-off error in the LS detector.

Next we always omit $O \left( \varepsilon^2 \right)$ terms, where $\varepsilon = u / \sqrt{3}$ is the scaled machine precision. We scale it so that $\varepsilon^2$ is the average squared round-off error in case of rounding to nearest and uniform distribution of the truncated part. So, 
\[
\varepsilon \leqslant 2^{-11}/\sqrt{3}
\]
for half precision, and 
\[
\varepsilon \leqslant 2^{-24}/\sqrt{3}
\]
for single precision. Although this assumption about round-off distribution is not true in general, it is usually close to the truth and is often used, for example, when estimating quantization errors \cite{QuantBook}. Rounding errors are also assumed to be independent random numbers, although, again, strictly they are not independent.

We use $\sim$ or $\lesssim$ symbols to omit expectations, constant factors, and lower-order terms. In particular, for any matrices $A$ and $B$ of the same size by definition
\[
  \left\| A \right\|_F \lesssim \left\| B \right\|_F \Leftrightarrow \mathbb{E} \left\| A \right\|_F^2 \leqslant {\rm const} \cdot \mathbb{E} \left\| B \right\|_F^2.
\]
%where $\left\| \cdot \right\|$ is some matrix norm (which will be specified in each case). 
%It is also possible to use $\sim$ notation even if there is no randomness. In this case we can easily calculate $\mathbb{E} \left\| A \right\|^2 = \left\| A \right\|^2$ and $\mathbb{E} \left\| B \right\|^2 = \left\| B \right\|^2$.

The rounding errors (i.e., errors coming from inexact operations and inexact representation of numbers) in $A$ will be denoted by $\Delta A$, and 
\[
\tilde A = A + \Delta A
\]
is the inexact version of $A$. This includes all errors coming from any operations involved in the computation of $A$, as well as any introduced backward errors (like in Theorem \ref{th:CholeskyOld}).

In \cite{HighamPHD} randomized approach was used to calculate errors of QR decomposition, using the following theorem, which describes the round-off error of an arbitrary scalar product. Here we simplify the formulation using our notation.

\begin{theorem}[ \cite{HighamPHD}]\label{th:scalar_orig}
Let $a,b \in \mathbb{C}^N$, then, under the assumption of random independent round-off errors,
\begin{equation}\label{eq:scalar_eq}
  \left| \Delta \left( a^H b \right) \right| \lesssim \sqrt{N} \varepsilon \left\| a \right\|_2 \left\| b \right\|_2.
\end{equation}
\end{theorem}

\begin{remark}
Errors of scalar products can be significantly decreased if one uses a binary tree structure to calculate them. In this case, one gets $\sqrt{\ln N}$ coefficient instead of $\sqrt{N}$.
%($\ln N$ if the round-offs are not independent). 
%Even using block summation can improve the coefficient several times: in blocks of $8$, for example, the coefficient is going to be $\sqrt{N/8 + 8}$ ($N/8 + 8$ if the round-offs are not independent). Note also, that one can try to introduce an analogue of stochastic round-off explicitly (for example, by setting the last mantissa bit to 1 after each operation, if rounding to zero is used).
\end{remark}

In the case of Cholesky decomposition, however, we will need a stronger result. This is because even before we do Cholesky decomposition, we need to calculate $A = H^H H$ matrix product, which is a combination of scalar products. This product already leads to a huge initial error $\Delta \left( H^H H \right)$, which we will need to estimate accurately.

It is enough to focus on the round-off error in $A$ because the solution error will then be proportional to 
\[
{\rm cond}_2 \left( A \right) = {\rm cond}_2^2 \left( H \right).
\]
On the other hand, errors coming from the inexact computation of $L^{-1}$ or from $WY$ multiplication \eqref{eq:WY} only lead to errors, proportional to the first power ${\rm cond}_2 \left( H \right)$. Corresponding classical and randomized bounds can be found in \cite{accuracy} and \cite{HighamPHD} respectively.

\section{Better scalar product bound}

To prove a better scalar product bound, we first require some classic results about matrix volume.

\begin{definition}\label{def:vol}
Volume $\cV \left( A \right)$ of a matrix $A$ is
\[
  \cV \left( A \right) = \sqrt{\max \left( \det A^H A, \det A A^H \right)}.
\]
\end{definition}

\begin{theorem}[Binet-Cauchy]
Let $A_{\mathcal{I},:}, B_{\mathcal{I},:} \in \mathbb{C}^{N \times N}$ be square submatrices of $A, B \in \mathbb{C}^{M \times N}$, corresponding to the set of row indices $\mathcal{I}$ of cardinality $N$. Then
\[
  \det A^H B = \sum\limits_{\mathcal{I}} \det A_{\mathcal{I},:}^H B_{\mathcal{I},:}.
\]
Consequently,
\begin{equation}\label{eq:cVbinet}
  \cV^2 \left( A \right) = \sum\limits_{\mathcal{I}} \cV^2 \left( A_{\mathcal{I},:} \right).
\end{equation}
\end{theorem}
\begin{cons}\label{cor:binet}
Let $A$ be a random matrix with permutation invariant distribution (random unitary matrices and their submatrices are such examples). Then
\[
  \mathbb{E} \cV^2 \left( A \right) = C_M^N \mathbb{E} \cV^2 \left( A_{\mathcal{I},:} \right)
\]
as there are $C_M^N$ submatrices in the sum \eqref{eq:cVbinet}, and all of them are equally distributed because of permutation invariance.
\end{cons}

\begin{proposition}\label{prop:scalar}
Let $u, v \in \mathbb{C}^N$ be random unit vectors satisfying $u^H v = {\rm const}$. Let $a = u \left\| a \right\|_2$, $b = v \left\| b \right\|_2$, then
%\[
%  \left| \Delta \left( a^H b \right) \right| \sim \sqrt{\frac{N+1}{2}} \varepsilon \left| a^H b \right| + \sqrt{\frac{N+2}{4(N-1)}} \varepsilon \sqrt{\left\| a \right\|_2^2 \left\| b \right\|_2^2 - \left| a^H b \right|^2} \sim \sqrt{N} \varepsilon \left| a^H b \right| + \varepsilon \left\| a \right\|_2 \left\| b \right\|_2.
%\]
\begin{equation}\label{eq:aHb}
  \left| \Delta \left( a^H b \right) \right| \sim \sqrt{N} \varepsilon \left| a^H b \right| + \varepsilon \left\| a \right\|_2 \left\| b \right\|_2.
\end{equation}
\end{proposition}
\begin{proof}
Let $b_{\parallel} = a \left( a^H b \right) / \left\| a \right\|_2^2$ and $b_{\bot} = b - b_{\parallel}$. Then, using orthogonality and round-off error independence assumption,
\[
\begin{aligned}
  \mathbb{E} \left| a^H b + \Delta \left( a^H b \right) \right|^2 & = \mathbb{E} \left| a^H b_{\parallel} + \Delta \left( a^H b_{\parallel} \right) \right|^2 \\
  & + \mathbb{E} \left| a^H b_{\bot} + \Delta \left( a^H b_{\bot} \right) \right|^2,
\end{aligned}
\]
where $\Delta \left( a^H b_{\parallel} \right)$ and $\Delta \left( a^H b_{\bot} \right)$ are the errors of operations during the computation of $a^H b_{\parallel}$ and $a^H b_{\bot}$ respectively.

The first term calculates the inexact value of 
\[
\sum\limits_{i=1}^N \left| a_i \right|^2 = \left\| a \right\|_2^2.
\]
with elements of the sum premultiplied by $a^H b / \left\| a \right\|_2^2$. Let us calculate the increase in error, when the $k$-th term of the sum is added:
\begin{equation}\label{eq:Deltaa2}
  \left| \Delta \left( \sum\limits_{i=1}^k \left| a_i \right|^2 \right) \right| \lesssim \left| \Delta \left( \sum\limits_{i=1}^{k-1} \left| a_i \right|^2 \right) \right| + \sqrt{2} \varepsilon \sum\limits_{i=1}^{k} \left| a_i \right|^2,
\end{equation}
where we accounted for error independence and the fact that summation contains two terms (for real and imaginary parts individually), leading to a coefficient $\sqrt{2}$ (for fused add-multiply).

Taking the expectation of the sum in the r.h.s. of \eqref{eq:Deltaa2}, we find
\[
  \left| \Delta \left( \sum\limits_{i=1}^k \left| a_i \right|^2 \right) \right| \lesssim \left| \Delta \left( \sum\limits_{i=1}^{k-1} \left| a_i \right|^2 \right) \right| + \sqrt{2} \varepsilon \frac{k}{N} \left\| a \right\|_2^2,
\]
so in total
\[
\begin{aligned}
  \left| \Delta \left( a^H b_{\parallel} \right) \right| & \lesssim \frac{\left| a^H b \right|}{\left\| a \right\|_2^2} \sqrt{\sum\limits_{k=1}^N 2 \varepsilon^2 \frac{k^2}{N^2} \left\| a \right\|_2^4} \\
  & = \sqrt{\frac{(N+1)(2N+1)}{3N}} \varepsilon \left| a^H b \right|
  \sim \sqrt{N} \varepsilon \left| a^H b \right|.
\end{aligned}
\]
Thus, we obtained the first term in \eqref{eq:aHb}.
\end{proof}
Next, let us consider $\Delta \left( a^H b_{\bot} \right)$. Here the scalar product $a^H b_{\bot} = 0$, and $b_{\bot}$ is still a random vector with fixed $\left\| b_{\bot} \right\|_2^2 = \left\| b \right\|_2^2 - \left| a^H b \right|^2$. Since $u$ and $v_{\bot} = b_{\bot} / \left\| b_{\bot} \right\|_2$ are random unit orthogonal vectors, we can consider them as two consecutive columns of a unitary matrix. Consider its $k \times 2$ submatrix $A = \left[ u_{1:k} \; v_{\bot,1:k} \right]$. Its expected squared volume (computed by definition \ref{def:vol}) is
\[
\begin{aligned}
  \mathbb{E} \cV^2 \left( A \right)  = \mathbb{E} \left( \left\| u_{1:k} \right\|_2^2 \left\| v_{\bot,1:k} \right\|_2^2 \right) 
   -  \mathbb{E} \left| u_{1:k}^H v_{\bot,1:k} \right|^2 \\
   = \mathbb{E}_u \left[ \left\| u_{1:k} \right\|_2^2 \right] \mathbb{E}_{v_{\bot}} \left[ \left\| v_{\bot,1:k} \right\|_2^2 | u \right]
   -  \mathbb{E} \left| u_{1:k}^H v_{\bot,1:k} \right|^2.
\end{aligned}
\]

Conditional expectation $\mathbb{E}_{v_{\bot}} \left[ \left\| v_{\bot,1:k} \right\|_2^2 | u \right]$ can be easily calculated: there are in total $N-1$ columns in a unitary matrix (apart from $u$), total Frobenius norm of $k$ corresponding rows is equal to $k - \left\| u_{1:k} \right\|_2^2$, so the expected norm of each of these $N-1$ columns is
\[
  \mathbb{E}_{v_{\bot}} \left[ \left\| v_{\bot,1:k} \right\|_2^2 | u \right] = \left( N - \left\| u_{1:k} \right\|_2^2 \right) / \left( N - 1 \right).
\]
Thus, we continue
\begin{equation}\label{eq:cVA}
\begin{aligned}
  &\mathbb{E} \cV^2 \left( A \right) = \\
  &\mathbb{E} \left[ \left\| u_{1:k} \right\|_2^2 \left( k - \left\| u_{1:k} \right\|_2^2 \right) \right] / (N-1) 
   -  \mathbb{E} \left| u_{1:k}^H v_{\bot,1:k} \right|^2 \\
  & \leqslant \mathbb{E} \left\| u_{1:k} \right\|_2^2 \left( k - \mathbb{E} \left\| u_{1:k} \right\|_2^2 \right) / (N-1)
   -  \mathbb{E} \left| u_{1:k}^H v_{\bot,1:k} \right|^2 \\
  & = \frac{k^2}{N^2} -  \mathbb{E} \left| u_{1:k}^H v_{\bot,1:k} \right|^2.
\end{aligned}
\end{equation}
On the other hand, using corollary \ref{cor:binet}, we obtain
\begin{equation}\label{eq:cVuv}
  1 = \cV^2 \left( \left[ u \; v_{\bot} \right] \right) = \frac{C_N^2}{C_k^2} \mathbb{E} \cV^2 \left( A \right) = \frac{N(N-1)}{k(k-1)} \cV^2 \left( A \right).
\end{equation}
Combining \eqref{eq:cVA} and \eqref{eq:cVuv} together, we obtain the inequality
\[
  \mathbb{E} \left| u_{1:k}^H v_{\bot,1:k} \right|^2 \leqslant \frac{k^2}{N^2} - \frac{k(k-1)}{N(N-1)} \leqslant \frac{k}{N^2}.
\]
Now we can estimate errors in the partial sums the same way as in equation \eqref{eq:Deltaa2}:
\[
\begin{aligned}
  \left| \Delta \left( \sum\limits_{i=1}^k {\rm conj} \left( a_i \right) b_{\bot,i} \right) \right| & \lesssim \left| \Delta \left( \sum\limits_{i=1}^{k-1} {\rm conj} \left( a_i \right) b_{\bot,i} \right) \right| \\
  & + \sqrt{2} \varepsilon \left| \sum\limits_{i=1}^k {\rm conj} \left( a_i \right) b_{\bot,i} \right| \\
   \sim \left| \Delta \left( \sum\limits_{i=1}^{k-1} {\rm conj} \left( a_i \right) b_{\bot,i} \right) \right| 
  & + \sqrt{2} \varepsilon \frac{\sqrt{k}}{N} \left\| a \right\|_2 \left\| b_{\bot} \right\|_2.
\end{aligned}
\]
This summation can be performed from both sides ($1$ to $N/2$ and $N$ to $N/2$) for a better estimate. Then we get
\[
\begin{aligned}
  & \left| \Delta \left( a^H b_{\bot} \right) \right| \lesssim \\
  & \lesssim \sqrt{\sum\limits_{k=1}^{\left\lfloor N/2 \right\rfloor} 2 \varepsilon^2 \frac{k}{N^2} \left\| a \right\|_2^2 \left\| b_{\bot} \right\|_2^2 + \!\!\!\! \sum\limits_{k=1}^{\left\lfloor (N+1)/2 \right\rfloor} 2 \varepsilon^2 \frac{k}{N^2} \left\| a \right\|_2^2 \left\| b_{\bot} \right\|_2^2} \\
  & \leqslant \frac{N+1}{\sqrt{2}N} \varepsilon \left\| a \right\|_2 \left\| b_{\bot} \right\|_2 
  \lesssim \varepsilon \left\| a \right\|_2 \left\| b \right\|_2,
\end{aligned}
\]
which is the second term in \eqref{eq:aHb}.

In case there is no fused add-multiply, there is an additional error coming from the round-off of each individual product. We can estimate it as
\[
\begin{aligned}
  \sum\limits_{i=1}^N \left| \Delta \left( {\rm conj} \left( a_i \right) b_i \right) \right| & \sim \sum\limits_{i=1}^N \varepsilon \left| {\rm conj} \left( a_i \right) b_{i, \parallel} \right| \\
  & + \varepsilon \sum\limits_{i=1}^N \left| {\rm conj} \left( a_i \right) b_{i, \bot} \right| \\
  & \sim \frac{\varepsilon}{\sqrt{N}} \left| a^H b \right| + \frac{\varepsilon}{\sqrt{N}} \left\| a \right\|_2 \left\| b_{\bot} \right\|_2,
\end{aligned}
\]
which is lower than that of previous sources.

%\begin{remark}
%RANDSVD ensemble (see definition \ref{def:randsvd} below) does not guarantee that assumptions of the proposition above are satisfied, but we will use it anyway. Also, we hypothesize that this proposition should still hold if $\left| a_i b_i \right|$ are ordered from lowest to highest (and still have random phases): in this case any intermediate sum will (on average) be lower than the case, when all elements are equal in absolute value.
%\end{remark}

\section{More assumptions}\label{sec:assum}

To get more accurate results for the Cholesky decomposition error estimate, we require more assumptions. First, we are going to consider $X$ to have uniformly random direction in $\mathbb{C}^N$, and $H$ taken (independently) from the RANDSVD ensemble.

\begin{definition}\label{def:randsvd}
$A \sim RANDSVD \left( \Sigma \right)$ for diagonal matrix $\Sigma$ if
\[
  A = U \Sigma V,
\]
where $U \in \mathbb{C}^{M \times \min \left( M, N \right)}$ and $V \in \mathbb{C}^{\min \left( M, N \right) \times N}$ are the first submatrices of independent random unitary matrices (with Haar measure).
\end{definition}

This is a natural way to construct an ensemble of matrices for any arbitrary set of singular values. Naturally, channel matrices in practice are not taken from this ensemble. Nevertheless, any arbitrary matrix lies in the ensemble, corresponding to its singular values, so when we prove some bound for the RANDSVD ensemble, it already means that there are few matrices (a subset of low measure) that break this bound. As channel matrices and round-offs in them are usually far from the worst-case bound of Theorem \ref{th:CholeskyOld}, this is a reasonable replacement for the unknown ensemble of arbitrary channel matrices.

Next, we are going to use the following assumptions:
\begin{enumerate}
\item We replace round-off errors with Gaussian random variables with the same expectation and variance. %We are not going to keep track of the constant factors, and we will use the square root of variance as an ``average''.

\item Similarly, matrices of errors (which we will denote by $\Delta A$) will be replaced with random Gaussian matrices with the same average Frobenius norm.
%We can also replace lower or upper triangular matrices by full matrices, as that only increases the error and by at most a constant factor. In other words, if we have some matrix $\Delta R$, which is random upper triangular, we can always add independent lower triangular error to it to make it square. We will denote random (normal) Gaussian matrix of size $M \times N$ by $\Omega_{M \times N}$.
\end{enumerate}

Replacing error matrices with random Gaussian matrices allows estimating errors using the following proposition, which we already successfully used in \cite{OsinskyRegularization}.

\begin{proposition}[\cite{GaussProduct}]\label{prop:gauss}
Let $\Omega_{M \times N} \in \mathbb{C}^{M \times N}$ be a random Gaussian matrix with independent entries. Then
\[
\begin{aligned}
  \left\| A \Omega_{M \times N} B \right\|_F & \sim \frac{\left\| A \right\|_F \left\| \Omega_{M \times N} \right\|_F \left\| B \right\|_F}{\sqrt{MN}} \\ 
  & \sim \left\| A \right\|_F \left\| B \right\|_F.
\end{aligned}
\]
\end{proposition}

\section{Cholesky decomposition error}

\begin{theorem}\label{th:Cholesky}
Under all the above assumptions, floating point Cholesky of positive definite $A \in \mathbb{C}^{N \times N}$, $A = H^H H$ is equivalent (when runs to completion) to exact Cholesky of $\tilde A = A + \Delta A$, where
\begin{equation}\label{eq:DeltaA}
  \Delta A = \Delta A_1 + \Delta A_1^H
\end{equation}
and
\begin{equation}\label{eq:DeltaA1}
  \left\| \Delta A_1 \right\|_F \lesssim \sqrt{N} \varepsilon \left\| A \right\|_F.
\end{equation}

Moreover, elements of the error bound for $\Delta A_1$ are independent and equally distributed.
\end{theorem}

%Note that we won't need to know the bound for $\Delta L$, since $\Delta A$ will be used directly.

First, we remind the pseudocode for Cholesky decomposition, which we will refer to.
\begin{algorithmic}[1]
\For{$j := 1$ \To $N$}
  \State $L_{jj} := \sqrt{A_{jj}}$ \label{line:Ljj}
  \State $L_{j+1:N,j} := A_{j+1:N,j} / L_{jj}$ \label{line:rescL}
  \State $A := A - L_{:,j} L_{:,j}^H$ \label{line:updateA}
\EndFor
\end{algorithmic}

The error of the square root in line \ref{line:Ljj} can be compensated by introducing a relative error of the order $2 \varepsilon$ into diagonal elements of $A$. This error, when summed over all diagonal elements, is bounded by $\lesssim 2 \varepsilon \left\| A \right\|_F$, which is lower than our estimate \eqref{eq:DeltaA1}. Since diagonal elements never increase during the decomposition and for RANDSVD $H$ are initially (on average) the same, the bound for different elements is the same.

Next, division in line \ref{line:rescL} can also be accounted for by introducing relative error $\varepsilon$ (error in $L_{jj}$ was already accounted for). This similarly leads to a bound $\lesssim \varepsilon \left\| A \right\|_F$, which is also lower than our estimate \eqref{eq:DeltaA1}. Here we use the fact that nondiagonal elements (on average) decrease after each Cholesky step and are initially (on average) the same for $H$ from RANDSVD.

Finally, line \ref{line:updateA} affects each element of $A$ at most $N$ times, and each time introduces relative error $\varepsilon$. Again, on average, the most error is introduced during the first step. Using the independence of round-offs, we get the total error as in \eqref{eq:DeltaA1}.

To turn error into random Gaussian matrices with independent entries, we use assumption 2 in section \ref{sec:assum} and extend both lower and upper triangular parts to full matrices $\Delta A_1$ and $\Delta A_1^H$, which leads to \eqref{eq:DeltaA}.

Note, however, that the error in Theorem \ref{th:Cholesky} is not the largest. A larger error comes from multiplication $A = H^H H$. From proposition \ref{prop:scalar} applied to each element of $H^H H$, we derive
\[
  \left\| \Delta \left( H^H H \right) \right\|_F \lesssim \sqrt{M} \varepsilon \left\| A \right\|_F + \left\| \Delta \hat A \right\|_F,
\]
where 
\[
\left| \Delta \hat A_{ij} \right| \lesssim \varepsilon \sqrt{A_{ii} A_{jj}},
\]
therefore
\[
\left\| \Delta \hat A \right\|_F \lesssim \sqrt{N} \varepsilon \left\| A \right\|_F
\]
and the dominant term is $\sqrt{M} \varepsilon \left\| A \right\|_F$.

\section{Effect of round-off error on the LS detector}

When considering the solution of the LS problem, condition numbers will appear. We remind their definition here.

\begin{definition}
\[
  {\rm cond}_2 \left( A \right) = \left\| A \right\|_2 \left\| A^+ \right\|_2.
\]
\[
  {\rm cond}_F \left( A \right) = \left\| A \right\|_F \left\| A^+ \right\|_F.
\]
They are connected as
\[
  {\rm cond}_2 \left( A \right) \leqslant {\rm cond}_F \left( A \right) \leqslant N {\rm cond}_2 \left( A \right).
\]
%or, more precisely, as
%\[
%  {\rm cond}_2 \left( A \right) + N - 2 + \frac{1}{{\rm cond}_2 \left( A \right)} \leqslant {\rm cond}_F \left( A \right) \leqslant \frac{N}{2} \left( {\rm cond}_2 \left( A \right) + 1 \right),
%\]
%where $A \in \mathbb{C}^{M \times N}$, $M \geqslant N$.
\end{definition}

%For exponential distribution with singular values $\sigma_k \left( A \right) = {\rm const} \cdot {\rm cond}_2^{(k-1)/(N-1)} \left( A \right)$, ${\rm cond}_F \left( A \right) = {\rm cond}_2 \left( A \right) \frac{1 - {\rm cond}_2^{-2N/(N-1)} \left( A \right)}{1 - {\rm cond}_2^{-2/(N-1)} \left( A \right)}$, ${\rm cond}_2 \left( A \right) > 1$. 
Since in practice, it is common to see close to the exponential distribution of singular values of $H$, we usually have ${\rm cond}_F \left( H \right) \sim {\rm cond}_2 \left( H \right)$. Although ${\rm cond}_F$ is not a standard notation, we still use it for completeness (in addition to standard ${\rm cond}_2$), since it matches our bounds more closely. 
%Moreover, ${\rm cond}_F$ is much easier to calculate in practice: calculation of $\left\| H \right\|_F$ costs $O(MN)$ (compared to cubic complexity of calculating spectral norm) and $\left\| H^+ \right\|_F$ can usually be bounded by the noise power $\sigma^2$ in the detector (see section \ref{sec:introduction}).

For vector $X$ we use an assumption $\left\| X \right\|_2 = 1$ so that we get a relative error. We directly substitute error for $\Delta A$ ($A = H^H H$) into $\Delta X \sim \left( A^{-1} - \tilde A^{-1} \right) H^H Y$, leading to (using proposition \ref{prop:gauss}):
\begin{equation}\label{eq:final}
\begin{aligned}
  \left\| \Delta X \right\|_F & \sim \left\| \left( A + \Delta A \right)^{-1} H^H Y - X \right\|_F \\
  & = \left\| \left( A + \Delta A \right)^{-1} H^H Y - A^{-1} H^H Y \right\|_F \\
  & \sim \left\| A^{-1} \Delta A A^{-1} H^H Y \right\|_F \\
  & = \left\| A^{-1} \Delta A A^{-1} H^H H X \right\|_F \\
  & \sim \frac{1}{N} \left\| A^{-1} \right\|_F \left\| \Delta A \right\|_F \left\| X \right\|_2 \\
  & \lesssim \frac{\sqrt{M}}{N} \varepsilon \left\| A^{-1} \right\|_F \left\| A \right\|_F \\
  & = \frac{\sqrt{M}}{N} \varepsilon \left\| \left(H^H H \right)^{-1} \right\|_F \left\| H^H H \right\|_F \\
  & = \frac{\sqrt{M}}{N} \varepsilon {\rm cond}_F \left( H^H H \right) 
  \leqslant \sqrt{M} \varepsilon {\rm cond}_2^2 \left( H \right),
\end{aligned}
\end{equation}
where we also used the Neumann series
\[
  \tilde A^{-1} = \left( A + \Delta A \right)^{-1} = A^{-1} \Delta A A^{-1} + \ldots
\]
Using Theorem \ref{th:CholeskyOld} would instead lead to an estimate at least $N$ times higher.

\section{Simulation results}

We performed numerical simulations for RANDSVD matrices with exponential distribution of singular values (default mode of the gallery('randsvd', ...) in Matlab) and compared the observed round-off errors of half precision arithmetic with the theoretical bound \eqref{eq:final}. In Figure \ref{Error_32x32} we plot the average error for randsvd matrices of size $M=N=32$ with its variance, and a theoretical upper bound. By definition, we bound not the worst case, but the average expected error, and, indeed, the theoretical orange curve always lies above but close to the blue curve, which shows the average error of the LS solution. In Figure \ref{Error_64x12_real} we plot a more realistic case $M = 64$ antennas and $N=12$ users. In addition to the distribution of error for randsvd matrices, the figure also shows the distribution for realistic channels, constructed with the Quadriga 2.0 generator \cite{Quadriga}. Here theoretical bound is also close, but larger than the average error. The fact that real channels show the same error distribution as the corresponding randsvd matrices also justifies our assumption of using randsvd ensemble to predict roundoff errors.
%Our upper bound is $\frac{\sqrt{M}}{N} \varepsilon {\rm cond}_F \left( H^H H \right)$, while actual error is closer to the sum $\frac{\sqrt{M}}{N} \varepsilon {\rm cond}_F \left( H^H H \right) + \frac{1}{\sqrt{N}} \varepsilon {\rm cond}_F \left( H^H H \right)$ with different constants. Consequently, our theoretical bound is closer to practice in Figure \ref{Error_32x32}, where $M = N$ and both terms are of the same order, while in Figure \ref{Error_64x12} the second term $\frac{1}{\sqrt{N}} \varepsilon {\rm cond}_F \left( H^H H \right)$ can be neglected compared to the first.

% \begin{figure}[t!]
% \centering
% \includegraphics[width=0.98\columnwidth]{figures/Cholesky_64x12.eps}
% \caption{Dependency of Cholesky round-off error on condition number for 64x12 matrix H.
% }
% \label{Error_64x12}
% \end{figure}

\begin{figure}[t!]
\centering
\includegraphics[width=1.00\columnwidth]{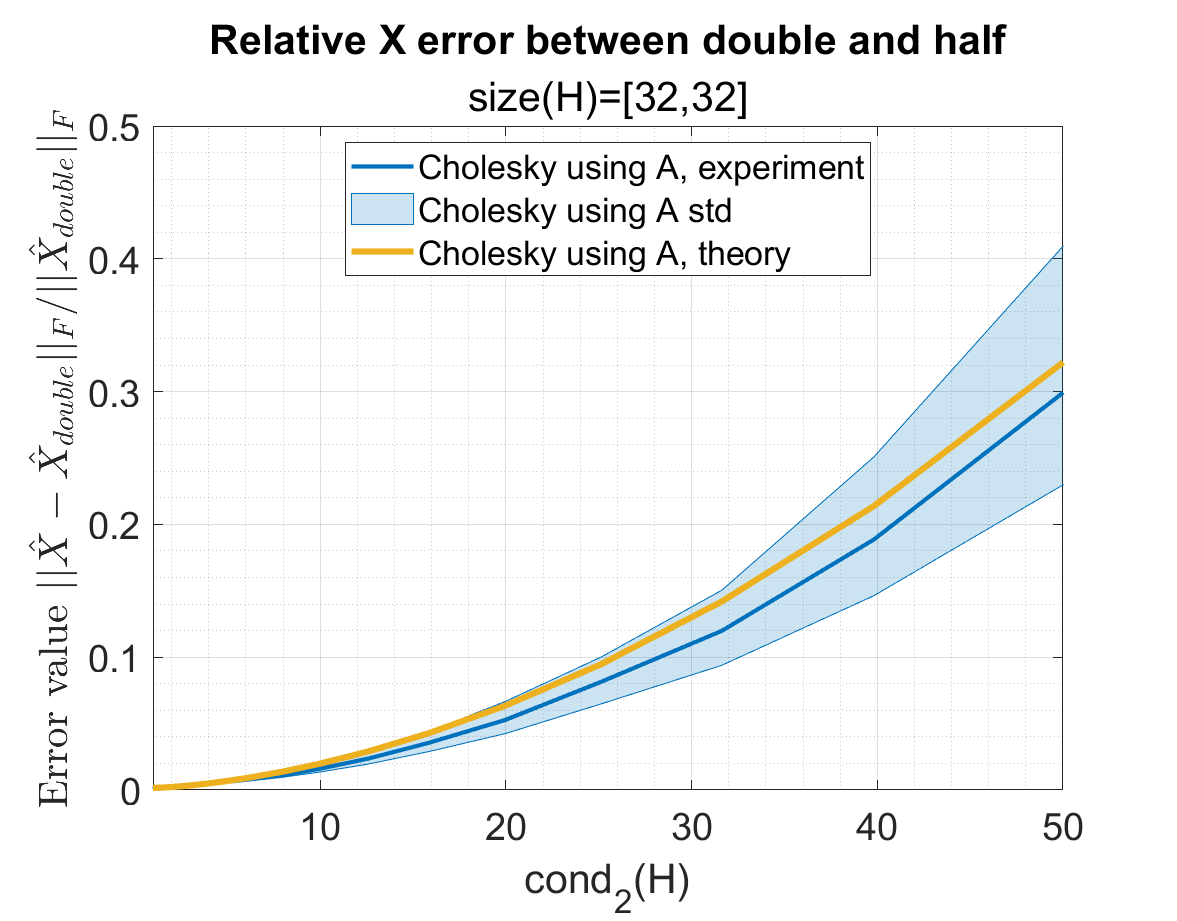}
\caption{Dependency of Cholesky round-off error on condition number for 32x32 matrix H. 
}
\label{Error_32x32}
\end{figure}

\begin{figure}[t!]
\centering
\includegraphics[width=1.00\columnwidth]{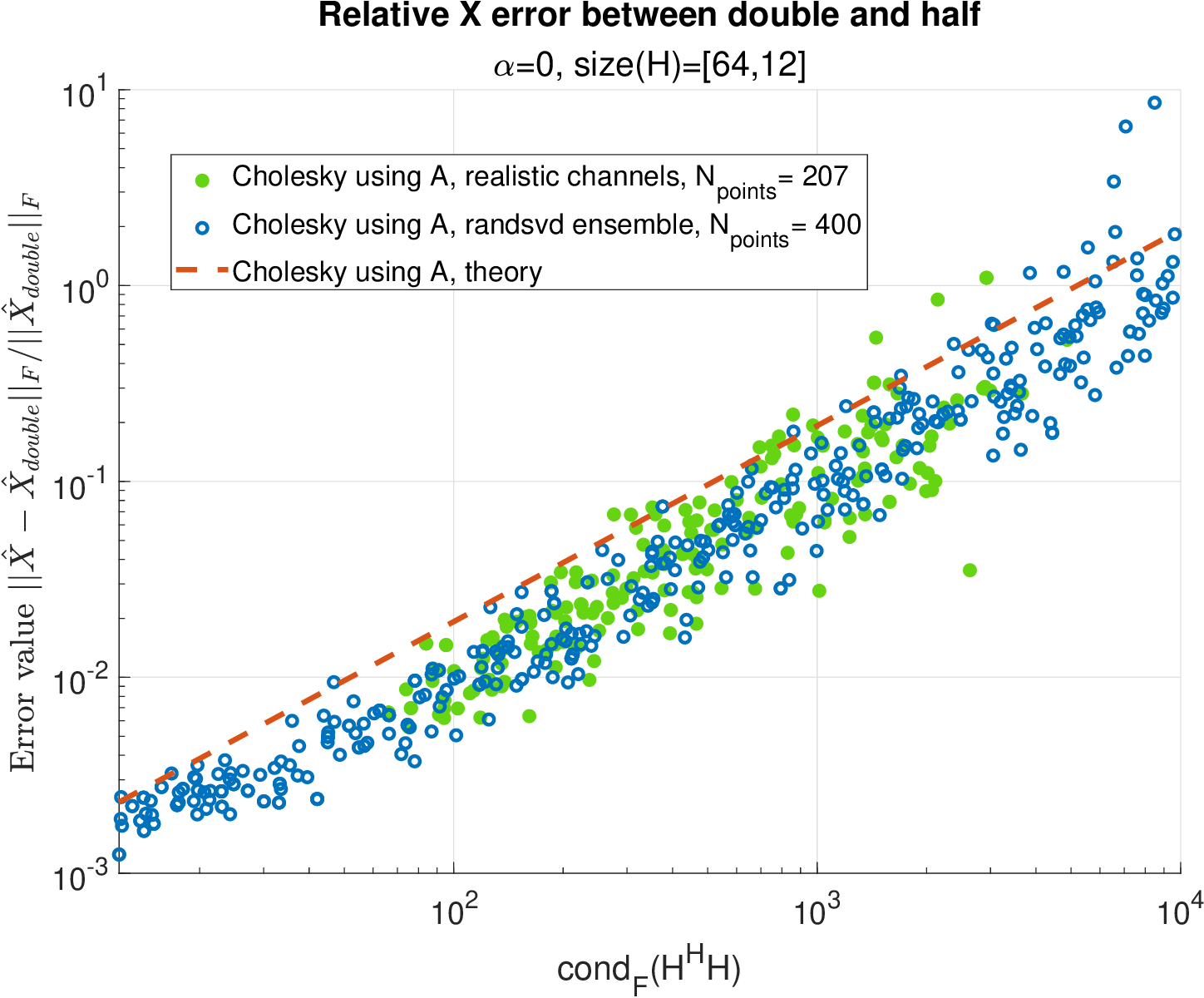}
\caption{Dependency of Cholesky round-off error on condition number for 64x12 matrices H from RANDSVD ensemble and QuaDRiGa simulations. 
}
\label{Error_64x12_real}
\end{figure}

%IMPORTANT???!!!
%In final version it is better to generate realistic channels from Quadriga (referencing parameters from an arbitrary paper to avoid length growth; for example, from one of our previous Cholesky ones). And plot them as a dot cloud (140 dots) with cond_F instead of cond_2 on the figures.

\section{Conclusion}\label{sec:conclusion}

The proposed method allows to analyze the resulting error from the low-precision Cholesky decomposition in the linear LS solution. The numerical results demonstrate the theoretical error to be close to the observed one for half-precision format. Moreover, the bound exceeds practical errors by $<\!1\,$dB that allows choosing the required bitwidth with accuracy up to $<\!1$ bit. 
%CHANGED: analytical -> theoretical; actual -> observed.
%Simulations demonstrate analytically calculated MIMO detection error caused by low-bitwidth Cholesky decomposition stands close to obtained through simulations. 
Therefore, when another error (induced by inaccurate $Y$ or $H$, dominates it, the bitwidth in the Cholesky decomposition can be lowered without performance loss. 
%Therefore, one can compare round-off error with errors from other sources (inaccurate $R_{UU}$ calculation, poor channel estimation, or other) and when it does not dominate, reduce the bitwidth.

\section{Acknowledgment}

The authors acknowledge the use of Zhores \cite{Zhores2019} for obtaining the results presented in this paper.

\bibliographystyle{ieeetr}
\bibliography{main}

\end{document}